\documentclass[10pt]{article}
\usepackage{}

\usepackage{pgf}
\usepackage[square, comma, sort&compress, numbers]{natbib}
\usepackage[english]{babel}
\usepackage[utf8]{inputenc}
\usepackage{datetime}
\usepackage{geometry}
\usepackage{array}
\usepackage{xspace}

\usepackage{amsmath,amsfonts,amssymb,amsopn,amsthm,amscd}
\theoremstyle{plain}
\usepackage{graphicx} 
\usepackage{epstopdf}
\usepackage{enumerate}
\allowdisplaybreaks[4]

\def\Box{\vcenter{\vbox{\hrule\hbox{\vrule
     \vbox to 8.8pt{\hbox to 10pt{}\vfill}\vrule}\hrule}}}

\usepackage{indentfirst}
\newcommand{\Ff}{{\mathbb F}}

\newcommand\cc{{\mathcal C}}        %

\newcommand\cF{{\mathbf c}}

\def\Tr{\operatorname{Tr}}

\newtheorem{thm}{Theorem}[section]
\newtheorem{lem}[thm]{Lemma}

\def\Tr{\operatorname{Tr}}

\begin{document}

%

\title{Six constructions of asymptotically optimal codebooks via the character sums
}

\author{Wei Lu$^1$, Xia Wu$^{1*}$, Xiwang Cao$^{2}$, Ming Chen$^3$
}

\maketitle

\let\thefootnote\relax\footnotetext{This work was supported by the National Natural Science Foundation of China (Grant No. 11801070, 11771007, 61572027) and the Basic Research
Foundation (Natural Science).}

\let\thefootnote\relax\footnotetext{$^1$School of Mathematics, Southeast University, Nanjing 210096, China. (Email: luwei1010@seu.edu.cn)}

\let\thefootnote\relax\footnotetext{$^1$School of Mathematics, Southeast University, Nanjing 210096, China. (Email:wuxia80@seu.edu.cn)}

\let\thefootnote\relax\footnotetext{$^2$Department of Math, Nanjing University of Aeronautics and Astronautics, Nanjing 211100, China. (Email: xwcao@nuaa.edu.cn)}

\let\thefootnote\relax\footnotetext{$^3$School of Information Science and Engineering, Southeast University, Nanjing 210096, China. (Email: chenming@seu.edu.cn)}

\let\thefootnote\relax\footnotetext{$^*$Corresponding author. (Email: wuxia80@seu.edu.cn)}

\begin{abstract}
 In this paper, using additive characters of finite field, we find a codebook which is equivalent to the measurement matrix in \cite{MM}. The advantage of our construction is that it can be generalized naturally to construct the other five classes of codebooks using additive and multiplicative characters of finite field.  We determine the maximal cross-correlation amplitude of these codebooks by the properties of characters and character sums. We prove that all the codebooks we constructed are asymptotically optimal with respect to the Welch bound. The parameters of these codebooks are new.

{\bf Keywods}: Codebook, asymptotic optimality, Welch bound, Gauss sum, Jacobi sum.

\end{abstract}

\section{Introduction}
An $(N, K)$ codebook $\cc=\{\mathbf{c}_0,\mathbf{c}_1, ...,\mathbf{c}_{N-1}\}$ is a set of $N$ unit-norm complex vectors  $\mathbf{c}_i \in\mathbb{C}^K$ over an alphabet \textit{A}, where $i=0, 1, \ldots, N-1$. The size of \textit{A} is called the alphabet size of $\cc$. As a performance measure of a codebook in practical applications, the maximum magnitude of inner products between a pair of distinct vectors in $\cc$ is defined by
$$
I_{max}(\cc)=\underset{0\leq i\neq j \leq N-1}{\max}|\mathbf{c}_i\mathbf{c}_j^H|,
$$
where $\mathbf{c}_j^H$ denotes the conjugate transpose of the complex vector $\mathbf{c}_j$.
To evaluate an $(N, K)$ codebook $\cc$, it is important to find the minimun achievable $I_{max}(\cc)$ or its lower bound. The Welch bound \cite{Welch} provides a well-known lower bound on $I_{max}(\cc)$,
$$
I_{max}(\cc)\geq I_W=\sqrt{\frac{N-K}{(N-1)K}}.
$$
The equality holds if and only if for all pairs of $(i,j)$ with $i\neq j$
$$|\mathbf{c}_i\mathbf{c}_j^H|=\sqrt{\frac{N-K}{(N-1)K}}.$$
%

A codebook $\cc$ achieving the Welch bound equality is called a maximum-Welch-bound-equality (MWBE) codebook \cite{MWEB} or an equiangular tight frame \cite{frame}. MWBE codebooks are employed in various applications including code-division multiple-access(CDMA) communication systems \cite{Apl.CDMA}, communications \cite{MWEB}, combinatorial designs \cite{Apl.cd1,Apl.cd2,Apl.cd3}, packing \cite{Apl.pa}, compressed sensing \cite{Apl.cs}, coding theory \cite{Apl.code} and quantum computing \cite{Apl.qc}. To our knowledge, only the following  MWBE codebooks are presented as follows:

\begin{itemize}
\item $(N, N)$ orthogonal MWBE codebooks for any $N>1$ \cite{MWEB}, \cite{Apl.cd3};
\item $(N,N-1)$ MWBE codebooks for $N>1$ based on discrete Fourier transformation matrices \cite{MWEB}, \cite{Apl.cd3} or $m$-sequences \cite{MWEB};
\item $(N, K)$ MWBE codebooks from conference matrices \cite{Apl.pa}, \cite{MWEB1}, where $N=2K=2^{d+1}$ for a positive integer $d$ or $N=2K=p^d+1$  for a prime $p$ and a positive integer $d$;
\item $(N, K)$ MWBE codebooks based on $(N, K, \lambda)$ difference sets in cyclic groups \cite{Apl.cd3} and abelian groups \cite{Apl.cd1}, \cite{Apl.cd2};
\item $(N, K)$ MWBE codebooks from $(2, k, \nu)$-Steiner systems \cite{MWEB2};
\item $(N, K)$ MWBE codebooks depended on graph theory and finite geometries \cite{F1,F2,F3,F4}.
\end{itemize}

The construction of an MWBE codebook is known to be very hard in general, and the known classes of MWBE codebooks only exist for very restrictive $N$ and $K$. Many researches have been done instead to construct near optimal codebooks, i.e., codebook $\cc$ whose $I_{max}(\cc)$ nearly achieves  the Welch bound. In \cite{MWEB}, Sarwate gave some nearly optimal codebooks from codes and signal sets.  As an extension of the optimal codebooks based on difference sets, various types of near optimal codebooks based on almost difference sets, relative difference sets and cyclotomic classes were proposed, see \cite{Apl.cd1,NM2,NM5,NM6,zhou}. Near optimal codebooks constructed from binary row selection sequences were presented in \cite{NM1,NM3Yu}. In \cite{AL2,Luo,Luo2,Luo3}, some near optimal codebooks were constructed  via Jacobi sums and hyper Eisenstein sum.

In \cite{MM}, the authors combined a Reed-Solomon generator matrix with itself by the tensor product and employed this generated matrix to construct a complex measurement matrix. They proved that this matrix is asymptotically optimal according to the Welch bound. In this paper, we find a codebook which is equivalent to the measurement matrix in \cite{MM}. The codebook is actually the first construction in Section 3, using additive characters of finite field. The advantage of our construction is that it can be generalized naturally to construct the other five classes of codebooks using additive and multiplicative characters of finite field. We determine the maximal cross-correlation amplitude of these codebooks by the properties of characters and character sums.  All of these codebooks we constructed are near optimal according to the Welch bound. As a comparison, in Table 1, we list the parameters of some known classes of near optimal codebooks and those of the new ones.

This paper is organized as follows. In section 2, we recall some notations and basic results which will be needed in our discussion. In section 3, we present our six constructions of near optimal codebooks.  In section 4, we derive another family of codebooks, which are also near optimal. In section 5, we conclude this paper.

 \begin{table}[!htbp]
 \newcommand{\tabincell}[2]{\begin{tabular}{@{}#1@{}}#2\end{tabular}}
 \caption{\small The parameters of codebooks asymptotically meeting the Welch bound}
 \label{tabl1}
 \centering
  \setlength{\tabcolsep}{1mm}{
 \begin{tabular}{|c|c|c|c|}
   \hline
   Parameters $(N,K)$ & $I_{max}$ & References \\ \hline
  \tabincell{c}{ $(p^n,K=\frac{p-1}{2p}(p^n+p^{n/2})+1)$ \\
  with odd $p$ }& $\frac{(p+1)p^{n/2}}{2pK}$ & \cite{NM1} \\ \hline
  $(q^2, \frac{(q-1)^2}{2})$, $q=p^s$ with odd $p$ & $\frac{q+1}{(q-1)^2}$ & \cite{NM5} \\ \hline
  $q(q+4), \frac{q+1}{2}$, $q$ is a prime power & $\frac{1}{q+1}$ & \cite{Li} \\ \hline
    $(q, \frac{(q+3)(q+1)}{2})$, $q$ is a prime power  & $\frac{\sqrt{q}+1}{q-1}$ & \cite{Li} \\ \hline
   $(p^n-1,\frac{p^n-1}{2})$ with odd $p$ & $\frac{\sqrt{p^n}+1}{p^n-1}$ & \cite{NM3Yu} \\ \hline
   $(q^l+q^{l-1}-1,q^{l-1})$ for any $l>2$ & $\frac{1}{\sqrt{q^{l-1}}}$ & \cite{zhou} \\ \hline
    \tabincell{c} {$((q-1)^k+q^{k-1},q^{k-1}$),\\ for any $k>2$ and $q\geq4$} & $\frac{\sqrt{q^{k+1}}}{(q-1)^k+(-1)^{k+1}}$ & \cite{AL2} \\ \hline
  \tabincell{c} {$((q-1)^k+K,K$), for any $k>2$, \\ where $K=\frac{(q-1)^k+(-1)^{k+1}}{q}$ } & $\frac{\sqrt{q^{k-1}}}{K}$ & \cite{AL2} \\ \hline
   \tabincell{c}{$((q^s-1)^n+K,K)$, \\ for any $s>1$ and $n>1$, \\ where $K=\frac{(q^s-1)^n+(-1)^{n+1}}{q})$ } & $\frac{\sqrt{q^{sn+1}}}{(q^s-1)^n+(-1)^{n+1}}$ & \cite{Luo} \\ \hline
   \tabincell{c}{$((q^s-1)^n+q^{sn-1},q^{sn-1})$, \\ for any $s>1$ and $n>1$} & $\frac{\sqrt{q^{sn+1}}}{(q^s-1)^n+(-1)^{n+1}}$ & \cite{Luo} \\ \hline
   \tabincell{c}{$(q-1,\frac{q(r-1)}{2r})$, \\ $r=p^t, q=r^s$, with odd $p$ and $p\nmid s$} & $\frac{\sqrt{r}}{\sqrt{q}(\sqrt{r}-1)K}$ & \cite{Wu} \\ \hline
   \tabincell{c}{$(q^2,\frac{q(q+1)(r-1)}{2r})$, \\ $r=p^t, q=r^s$, with odd $p$} & $\frac{(r+1)q}{2rK}$ & \cite{Wu} \\ \hline
   \tabincell{c}{$(q^3,q^2)$ and $(q^3+q^2,q^2)$, \\ $q$ is a prime power} & $\frac{1}{q}$ & this paper \\ \hline
   \tabincell{c}{$((q-1)q^2,(q-1)q)$ and $(q^2-1,(q-1)q)$, \\ $q$ is a prime power} & $\frac{1}{q-1}$ & this paper \\ \hline
   \tabincell{c}{$((q-1)q^2,(q-1)^2)$ and $(q^3-2q+1,(q-1)^2)$, \\ $q$ is a prime power} & $\frac{q}{(q-1)^2}$ & this paper \\ \hline
   \tabincell{c}{$((q-1)^2q,(q-1)^2)$ and $(q^3-q^2-q+1,(q-1)^2)$, \\ $q$ is a prime power} & $\frac{q}{(q-1)^2}$ & this paper \\ \hline
   \tabincell{c}{$((q-1)^2q,(q-1)(q-2))$ and $(q^3-q^2-2q+2,(q-1)(q-2))$, \\ $q$ is a prime power} & $\frac{q}{(q-1)(q-2)}$ & this paper \\ \hline
    \tabincell{c}{$((q-1)^3,(q-2)^2)$ and $(q^3-2q^2-q+3,(q-2)^2)$, \\ $q$ is a prime power} & $\frac{q}{(q-2)^2}$ & this paper \\ \hline
 \end{tabular}}
 \end{table}

\section{Preliminaries}
In this section, we introduce some basic results on characters and character sums over finite fields, which will play important roles in the constructions of codebooks.

In this paper, we set $q$ be a power of a prime $p$, and $\Ff_q$ be a finite field with $q$ elements. For a set $E$, $\#E$  denotes the cardinality of $E$.

%
%

\subsection{characters over finite fields}
Let $\Ff_q$ be a finite field, in this subsection, we recall the definitions of the additive and multiplicative characters of $\Ff_q$.

For each $a\in \mathbb{F}_q$, an additive character of $\mathbb{F}_q$ is defined by the function $\chi_a(x)=\zeta_p^{\Tr_{q/p}(ax)}$, where $\zeta_p$ is a primitive $p-$th root of complex unity and $\Tr_{q/p}(\cdot)$ be the trace functions from $\Ff_q$ to $\Ff_p$. By the definition, $\chi_a(x)=\chi_1(ax)$. When $a=0$, we call $\chi_0$ the trivial additive character of $\mathbb{F}_q$. When  $a=1$, we call $\chi_1$ the canonical additive character of $\mathbb{F}_q$. Let $\widehat{\Ff_q}$ be the set of all additive characters of $\Ff_q$.The orthogonal relation of additive characters (see \cite{field}) is given by
$$
 \sum_{x\in \mathbb{F}_q}\chi_a(x)=\left\{
            \begin{array}{ll}
              q,& \hbox{if\ $a=0$,}\\
              0,& \hbox{otherwise.}
            \end{array}
          \right.
$$

As in \cite{field}, the multiplicative characters of $\mathbb{F}_q$ is defined as follows. For $j=0,1,...,q-2$, the functions $\varphi_j$ defined by
$$\varphi_j(\alpha^i)=\zeta_{q-1}^{ij},$$
are all the multiplicative characters of $\mathbb{F}_q$,
where $\alpha$ is a primitive element of $\mathbb{F}_q^*$, and $0\leq i\leq q-2$. If $j=0$, we have $\varphi_0(x)=1$ for any $x\in \mathbb{F}_q^*$, $\varphi_0$ is called the trivial multiplicative character of $\mathbb{F}_q$. Let $\widehat{\Ff_q^*}$ be the set of all the multiplicative characters of $\Ff_q^*$.

Let $\varphi$ be a multiplicative character of $\Ff_q$, the orthogonal relation of multiplicative characters (see \cite{field}) is given by
$$
 \sum_{x\in \mathbb{F}_q^*}\varphi(x)=\left\{
            \begin{array}{ll}
              q-1,& \hbox{if\ $\varphi=\varphi_0$,}\\
              0,& \hbox{otherwise.}
            \end{array}
          \right.
$$

\subsection{Character sums over finite fields}
\subsubsection{Gauss sum}
Let $\varphi$ be a multiplicative character of $\Ff_q$ and $\chi$  an additive character of $\Ff_q$. Then the Gauss sum over $\Ff_q$ is given by
\begin{equation*}
  G(\varphi,\chi)=\sum_{x\in \mathbb{F}_q^*}\varphi(x){\chi}(x).
\end{equation*}
For simplicity, we write $G(\varphi,\chi_1)$ over $\Ff_q$ simply as $g(\varphi)$.
It is easy to see the absolute value of $G(\varphi,\chi)$ is at most $q-1$, but is much smaller in general. The following lemma shows all the cases.
\begin{lem}{\rm \cite[Theorem 5.11]{field}}\label{gauss}
Let $\varphi$ be a multiplicative character and $\chi$ an additive character of $\Ff_q$. Then the Gauss sum $G(\varphi,\chi)$ over $\Ff_q$ satisfies
$$
 G(\varphi,\chi)=\left\{
            \begin{array}{ll}
              q-1, & \hbox{if\ $\varphi=\varphi_0$, $\chi=\chi_0$,} \\
              -1,& \hbox{if\ $\varphi=\varphi_0$, $\chi\neq\chi_0$,}\\
              0,& \hbox{if\ $\varphi\neq\varphi_0$, $\chi=\chi_0$.}
            \end{array}
          \right.
$$
For $\varphi\neq\varphi_0$ and $\chi\neq\chi_0$, we have $\left|G(\varphi,\chi)\right|=\sqrt{q}$.
\end{lem}

\begin{lem}{\rm \cite{field}}\label{gauss1}
Gauss sums for the finite field $\Ff_q$ satisfy the following property:

$G(\varphi,\chi_{ab})=\overline{\varphi}(a)G(\varphi,\chi_{b})$ for $a\in \Ff_q^*$, $b\in \Ff_q$, where $\overline{\varphi}$ denotes the complex conjugate of $\varphi$;


\end{lem}

\subsubsection{Jacobi sum}
The definition of a multiplicative character $\varphi$ can be extended as follows.

$$
\varphi(0)=\left\{
            \begin{array}{ll}
              1, & \hbox{if\ $\varphi=\varphi_0$,} \\
               0,& \hbox{if\ $\varphi\neq\varphi_0$.}
            \end{array}
          \right.
$$

Let $\varphi_1$ and $\varphi_2$ be multiplicative characters of $\Ff_q$. The sum
$$J(\varphi_1,\varphi_2)=\sum_{c_1+c_2=1, c_1,c_2\in \Ff_q}\varphi_1(c_1)\varphi_2(c_2)$$
is called a Jacobi sum in $\Ff_q$.

The values of Jacobi sums are given as follows.
\begin{lem}{\rm [\cite{field},Theorem 5.19,Theorem 5.20]}\label{jacobi}
For the values of Jacobi sums, we have the following results.

(1) If $\varphi_1$ and $\varphi_2$ are trivial, then $J(\varphi_1,\varphi_2)=q$.

(2) If one of the $\varphi_1$ and  $\varphi_2$ is trivial, the other is nontrivial, $J(\varphi_1,\varphi_2)=0$.

(3) If $\varphi_1$ and  $\varphi_2$ are both nontrivial and $\varphi_1\varphi_2$ is nontrivial, then $|J(\varphi_1,\varphi_2)|=\sqrt{q}$.

(4) If $\varphi_1$ and  $\varphi_2$ are both nontrivial and $\varphi_1\varphi_2$ is trivial, then $|J(\varphi_1,\varphi_2)|=1$.

\end{lem}

\subsection{A general construction of codebooks}
Let $D$ be a set and $K=\#D$. Let $E$ be a set of some functions which satisfy
$$f:D\rightarrow S,\ \ \hbox{where\  S\  is\  the\ unit\ circle.}$$
A general construction of codebooks is stated as follows in the complex plane,
$$\cc(D;E)=\{\cF_f:=\frac{1}{\sqrt{K}}(f(x))_{x\in D}, f\in E\}.$$

\section{six constructions of codebooks}
In this section, by multiplicative characters, additive characters, Gauss sums and Jacobi sums we construct new series of codebooks  which nearly meet the Welch bound.

\subsection{The first construction of codebooks}
Let

$$
D_1:=\{(x,y,z)\in \Ff_{q}\times\Ff_q\times\Ff_q: z=xy\}.
$$
Then $\#D_1=q^2$.

The codebook $\cc_1$ is constructed as
$$\cc_1:=\{\frac{1}{q}(\chi_a(x)\chi_b(y)\chi_c(z))_{(x,y,z)\in D_1}:a,b,c\in \Ff_q\}.$$
We can derive the following Theorem.

\begin{thm}\label{th01}
$\cc_1$ is a codebook with $N_1=q^3$, $K_1=q^2$ and $I_{max}(\cc_1)=\frac{1}{q}$.
\end{thm}

\begin{proof}
By the definition of $\cc_1$, it contains $q^3$ codewords of length $q^2$. Then it is easy to see $N_1=q^3$ and $K_1=q^2$. Let $\cF$ and $\cF'$ be any different codewords in $\cc_1$, $\cF=\frac{1}{q}(\chi_{a_1}(x)\chi_{b_1}(y)\chi_{c_1}(z))_{(x,y,z)\in D_1}$ and $\cF'=\frac{1}{q}(\chi_{a_2}(x)\chi_{b_2}(y)\chi_{c_2}(z))_{(x,y,z)\in D_1}$, where $a_1,a_2,b_1,b_2,c_1,c_2\in \Ff_q$. The correlation of $\cF$ and $\cF'$ is as follows. Then

\begin{eqnarray*}
&&K_1\cF\cF'^H\\
&=&\sum_{(x,y,z)\in D_1}\chi_{a_1}(x)\chi_{b_1}(y)\chi_{c_1}(z)\overline{\chi_{a_2}(x)\chi_{b_2}(y)\chi_{c_2}(z)}\\
&=&\sum_{(x,y,z)\in D_1}\chi((a_1-a_2)x+(b_1-b_2)y+(c_1-c_2)z)\\
&=&\sum_{x\in \Ff_q, y\in \Ff_q}\chi(ax+by+cxy),
\end{eqnarray*}
where $a=a_1-a_2, b=b_1-b_2, c=c_1-c_2$, since $\cF\neq\cF'$, not all of $a, b$ and $c$ equal to 0.

Then
\begin{eqnarray*}
&&K_1\cF\cF'^H\\
&=& \sum_{x\in \Ff_q}\chi(ax)\sum_{y\in \Ff_q}\chi(b+cx)y\\
&=&\sum_{x\in \Ff_q, b+cx=0}\chi(ax)q
\end{eqnarray*}
The last equation holds by the orthogonal relation of additive characters.

When $c\neq 0$, we have
$$K_1\cF\cF'^H=q\chi(-\frac{ab}{c}).$$

When $c=0$, $b\neq 0$, it is easy to get $K_1\cF\cF'^H=0$.

When $c=0$ and $b=0$, since $\cF\neq \cF'$, we get $a\neq 0$. Then
$$
K_1\cF\cF'^H=q\sum_{x\in \Ff_q}\chi(ax)=0$$
by the orthogonal relation of additive characters.

Therefore,
we have $$I_{max}(\cc_1)=max\{|\cF\cF'^H|:\cF,\cF'\in \cc, and\ \cF\neq\cF'\}=\frac{q}{K_1}=\frac{1}{q} .$$

\end{proof}

Using Theorem \ref{th01}, we can derive the ratio of $I_{max}(\cc_1)$ of the proposed codebooks to that of the MWBE codebooks and show the near-optimality of the proposed codebooks as in the following theorem.

\begin{thm}\label{th02}
Let $I_{W1}$ be the Welch bound equality, for the given $N_1$, $K_1$ in the current section. We have

$$\lim_{r\rightarrow\infty}\frac{I_{max}(\cc_1)}{I_{W1}}=1,$$
then the codebooks we proposed are near optimal.

\end{thm}
\begin{proof}
Note that $N_1=q^3$ and $K_1=q^2$. Then the corresponding Welch bound is
$$
I_{W1}=\sqrt{\frac{N_1-K_1}{(N_1-1)K_1}}=\sqrt{\frac{q^3-q^2}{(q^3-1)q^2}}=\sqrt{\frac{1}{q^2+q+1}},
$$
we have
$$\frac{I_{max}(\cc_1)}{I_{W1}}=\sqrt{\frac{q^2+q+1}{q^2}}.$$

It is obvious that $\lim_{q\rightarrow+\infty }\frac{I_{max}(\cc_1)}{I_{W1}}=1$. The codebook $\cc_1$ asymptotically meets the Welch bound. This completes the proof.
\end{proof}

In Table \ref{table2}, we provide some explicit values of the parameters of the codebooks we proposed for some given $q$, and corresponding numerical data of the Welch bound for comparison. The numerical results show  that the codebooks $\cc_1$ asymptotically meet the Welch bound.
\begin{table}[!htbp]
 \newcommand{\tabincell}[2]{\begin{tabular}{@{}#1@{}}#2\end{tabular}}
 \caption{\footnotesize Parameters of the $(N_1,K_1)$ codebook of Section IV}
 \label{table2}
\centering
 \setlength{\tabcolsep}{1mm}{
 \begin{tabular}{|c|c|c|c|c|c|}
   \hline
  $q$ & $N_1$ & $K_1$ & $I_{max}(\cc_1)$ & $I_{W1}$ & $\frac{I_{max}(\cc_1)}{I_{W1}}$ \\ \hline
 $3$& $27$& $9$& $0.3333$& $0.2774$& $1.2019$ \\ \hline
 $5$& $125$& $25$& $0.2000$& $0.1796$& $1.1136$ \\ \hline
 $13$& $2197$& $169$& $0.0769$& $0.0739$& $1.0406$ \\ \hline
 $49$& $117649$& $2401$& $0.0204$& $0.0202$& $1.0104$ \\ \hline
 $5^{3}$& $1953125$& $15625$& $0.0080$& $0.0080$& $1.0040$ \\ \hline
 $5^{4}$& $244140625$& $390625$& $0.0016$& $0.0016$& $1.0008$ \\ \hline
 $7^{4}$& $1.3841e+10$& $5764801$& $4.1649e-04$& $4.1641e-04$& $1.0002$ \\ \hline
 \end{tabular}}
 \end{table}

\subsection{The second construction of codebooks}
Let

$$
D_2:=\{(x,y,z)\in \Ff_{q}^*\times\Ff_q\times\Ff_q: z=xy\}.
$$
Then $\#D_2=(q-1)q$.

The codebook $\cc_2$ is constructed as
$$\cc_2:=\{\frac{1}{\sqrt{(q-1)q}}(\varphi(x)\chi_b(y)\chi_c(z))_{(x,y,z)\in D_2}:\varphi\in \widehat{\Ff_q^*},b,c\in \Ff_q\}.$$
We can derive the following Theorem.

\begin{thm}\label{th21}
$\cc_2$ is a codebook with $N_2=(q-1)q^2$, $K_2=(q-1)q$ and $I_{max}(\cc_2)=\frac{1}{q-1}$.
\end{thm}

\begin{proof}
By the definition of $\cc_2$, it contains $(q-1)q^2$ codewords of length $(q-1)q$. Then it is easy to see $N_2=(q^2-1)q$ and $K_2=(q-1)q$. Let $\cF$ and $\cF'$ be any different codewords in $\cc_2$,  $\cF=\frac{1}{\sqrt{K_2}}(\varphi_j(x)\chi_{b_1}(y)\chi_{c_1}(z))_{(x,y,z)\in D_2}$ and $\cF'=\frac{1}{\sqrt{K_2}}(\varphi_k(x)\chi_{b_2}(y)\chi_{c_2}(z))_{(x,y,z)\in D_2}$,  where $\varphi_j,\varphi_k\in \widehat{\Ff_q^*},b_1,b_2,c_1,c_2\in \Ff_q$. Then the correlation of $\cF$ and $\cF'$ is as follows.

\begin{eqnarray*}
&&K_2\cF\cF'^H\\
&=&\sum_{(x,y,z)\in D_2}\varphi_j(x)\chi_{b_1}(y)\chi_{c_1}(z)\overline{\varphi_k(x)\chi_{b_2}(y)\chi_{c_2}(z)}\\
&=&\sum_{x\in \Ff_q^*, y\in \Ff_q}\varphi_j\overline{\varphi_k}(x)\chi((b_1-b_2)y+(c_1-c_2)xy)\\
&=&\sum_{x\in \Ff_q^*, y\in \Ff_q}\varphi(x)\chi(by+cxy) \ \ (\hbox{where $\varphi=\varphi_j\overline{\varphi_k}$, $b=b_1-b_2$ and $c=c_1-c_2$})\\
&=&\sum_{x\in \Ff_q^*}\varphi(x)\sum_{y\in \Ff_q}\chi((b+cx)y)\\
&=&\sum_{x\in \Ff_q^*,b+cx=0}\varphi(x)q,
\end{eqnarray*}
The last equation holds by the orthogonal relation of additive characters.

When $b\neq 0$ and $c\neq 0$, we have
$$K_2\cF\cF'^H=q\varphi(-\frac{b}{c}).$$

When $b\neq 0$, $c=0$, or $b=0, c\neq 0$, it is easy to see $\cF\cF'^H=0$.

When $b=0$ and $c=0$, since $\cF\neq \cF'$, $\varphi$ is nontrivial. We have
$$K_2\cF\cF'^H=\sum_{x\in \Ff_q^*}\varphi(x)q=0,$$
by the orthogonal relation of multiplicative characters.

Therefore, we have $$I_{max}(\cc_2)=max\{|\cF\cF'^H|:\cF,\cF'\in \cc, and\ \cF\neq\cF'\}=\frac{q}{K_2}=\frac{1}{q-1} .$$
\end{proof}

Using Theorem \ref{th21}, we can derive the ratio of $I_{max}(\cc_2)$ of the proposed codebooks to that of the MWBE codebooks and show the near-optimality of the proposed codebooks as in the following theorem.

\begin{thm}\label{th22}
Let $I_{W2}$ be the Welch bound equality, for the given $N_2$, $K_2$ in the current section. We have

$$\lim_{r\rightarrow\infty}\frac{I_{max}(\cc_2)}{I_{W2}}=1,$$
then the codebooks we proposed are near optimal.

\end{thm}
\begin{proof}
Note that $N_2=(q-1)q^2$ and $K_2=(q-1)q$. Then the corresponding Welch bound is
$$
I_{W2}=\sqrt{\frac{N_2-K_2}{(N_2-1)K_2}}=\sqrt{\frac{(q-1)q^2-(q-1)q}{((q-1)q^2-1)(q-1)q}}=\sqrt{\frac{q-1}{q^3-q^2-1}},
$$
we have
$$\frac{I_{max}(\cc_2)}{I_{W2}}=\sqrt{\frac{q^3-q^2-1}{(q-1)^3}}.$$

It is obvious that $\lim_{q\rightarrow+\infty }\frac{I_{max}(\cc_2)}{I_{W2}}=1$. The codebook $\cc_2$ asymptotically meets the Welch bound. This completes the proof.
\end{proof}

In Table \ref{table3}, we provide some explicit values of the parameters of the codebooks we proposed for some given $q$, and corresponding numerical data of the Welch bound for comparison. The numerical results show  that the codebooks $\cc_2$ asymptotically meet the Welch bound.

 \begin{table}[!htbp]
 \newcommand{\tabincell}[2]{\begin{tabular}{@{}#1@{}}#2\end{tabular}}
 \caption{\footnotesize Parameters of the $(N_2,K_2)$ codebook of Section IV}
 \label{table3}
  \centering
 \setlength{\tabcolsep}{1mm}{
 \begin{tabular}{|c|c|c|c|c|c|}
   \hline
  $q$ & $N_2$ & $K_2$ & $I_{max}(\cc_2)$ & $I_{W2}$ & $\frac{I_{max}(\cc_2)}{I_{W2}}$ \\ \hline
 $3$& $18$& $6$& $0.5000$& $0.3430$& $1.4577$ \\ \hline
 $5$& $100$& $20$& $0.2500$& $0.2010$& $1.2437$ \\ \hline
 $13$& $2028$& $156$& $0.0833$& $0.0769$& $1.0831$ \\ \hline
 $49$& $115248$& $2352$& $0.0208$& $0.0204$& $1.0208$ \\ \hline
 $5^{3}$& $1937500$& $15500$& $0.0081$& $0.0080$& $1.0081$ \\ \hline
 $5^{4}$& $243750000$& $390000$& $0.0016$& $0.0016$& $1.0016$ \\ \hline
 $7^{4}$& $1.3836e+10$& $5762400$& $4.1667e-04$& $4.1649e-04$& $1.0004$ \\ \hline
 \end{tabular}}
 \end{table}

\subsection{The third construction of codebooks}
Let

$$
D_3:=\{(x,y,z)\in \Ff_{q}^*\times\Ff_q^*\times\Ff_q^*: z=xy\}.
$$
Then $\#D_3=(q-1)^2$.

The codebook $\cc_3$ is constructed as
$$\cc_3:=\{\frac{1}{q-1}(\chi_a(x)\chi_b(y)\varphi(z))_{(x,y,z)\in D_3}:a,b\in \Ff_q, \varphi\in \widehat{\Ff_q^*}\}.$$
We can derive the following Theorem.

\begin{thm}\label{th31}
 $\cc_3$ is a codebook with $N_3=q^2(q-1)$, $K_3=(q-1)^2$ and $I_{max}(\cc_3)=\frac{q}{(q-1)^2}$.
\end{thm}

\begin{proof}
By the definition of $\cc_3$, it contains $q^2(q-1)$ codewords of length $(q-1)^2$. Then it is easy to see $N_3=q^2(q-1)$ and $K_3=(q-1)^2$. Let $\cF$ and $\cF'$ be any different codewords in $\cc_3$, $\cF=\frac{1}{q-1}(\chi_{a_1}(x)\chi_{b_1}(y)\varphi_j(z))_{(x,y,z)\in D_3}$ and $\cF'=\frac{1}{q-1}(\chi_{a_2}(x)\chi_{b_2}(y)\varphi_k(z))_{(x,y,z)\in D_3}$, where $b_1,b_2,c_1,c_2\in \Ff_q, \varphi_j,\varphi_k\in \widehat{\Ff_q^*}$. Then the correlation of $\cF$ and $\cF'$ is as follows.

\begin{eqnarray*}
&&K_3\cF\cF'^H\\
&=&\sum_{(x,y,z)\in D_3}\chi_{a_1}(x)\chi_{b_1}(y)\varphi_j(z)\overline{\chi_{a_2}(x)\chi_{b_2}(y)\varphi_k(z)}\\
&=&\sum_{x,y\in \Ff_q^*}\chi((a_1-a_2)x+(b_1-b_2)y)\varphi_j\overline{\varphi_k}(xy)\\
&=&\sum_{x,y\in \Ff_q^*}\chi(ax+by)\varphi(xy) \ \ (\hbox{where $a=a_1-a_2$, $b=b_1-b_2$, and $\varphi=\varphi_j\overline{\varphi_k}$})\\
&=& \sum_{x\in \Ff_q^*}\chi_a(x)\varphi(x)\sum_{y\in \Ff_q^*}\chi_b(y)\varphi(y)\\
&=&G(\varphi,\chi_a)G(\varphi,\chi_b).
\end{eqnarray*}

When $\varphi$ is trivial, since $\cF\neq\cF'$,  we get $a\neq 0$ or $b\neq 0$,  by Lemma \ref{gauss}, we have
$$
 K_3\cF\cF'^H=G(\varphi,\chi_a)G(\varphi,\chi_b)=\left\{
            \begin{array}{ll}
              -(q-1),& \hbox{if\ $a=0,b\neq 0$\ or $b=0,a\neq 0$,}\\
              (-1)(-1),& \hbox{if\ $a,b\in \Ff_q^*$},
            \end{array}
          \right.
$$

When $\varphi$ is nontrivial, by Lemma \ref{gauss} and \ref{gauss1}, we have

$$
 K_3\cF\cF'^H=G(\varphi,\chi_a)G(\varphi,\chi_b)=\left\{
            \begin{array}{ll}
              0,& \hbox{if\ $a=0$\ or $b=0$,}\\
              \overline{\varphi}(ab)g^2(\varphi),& \hbox{if\ $a,b\in \Ff_q^*$}.
            \end{array}
          \right.
$$

Therefore, we have $$I_{max}(\cc_3)=max\{|\cF\cF'^H|:\cF,\cF'\in \cc, and\ \cF\neq\cF'\}=\frac{q}{K_3}=\frac{q}{(q-1)^2}.$$
\end{proof}

Using Theorem \ref{th31}, we can derive the ratio of $I_{max}(\cc_3)$ of the proposed codebooks to that of the MWBE codebooks and show the near-optimality of the proposed codebooks as in the following theorem.

\begin{thm}\label{th32}
Let $I_{W3}$ be the Welch bound equality, for the given $N_3$, $K_3$ in the current section. We have

$$\lim_{r\rightarrow\infty}\frac{I_{max}(\cc_3)}{I_{W3}}=1,$$
then the codebooks we proposed are near optimal.

\end{thm}
\begin{proof}
Note that $N_3=(q-1)q^2$ and $K_3=(q-1)^2$. Then the corresponding Welch bound is
$$
I_{W3}=\sqrt{\frac{N_3-K_3}{(N_3-1)K_3}}=\sqrt{\frac{(q-1)q^2-(q-1)^2}{((q-1)q^2-1)(q-1)^2}}=\sqrt{\frac{q^2-q+1}{(q^3-q^2-1)(q-1)}},
$$
we have
$$\frac{I_{max}(\cc_3)}{I_{W3}}=\frac{q}{q-1}\sqrt{\frac{q^3-q^2-1}{(q^2-q+1)(q-1)}}.$$

It is obvious that $\lim_{q\rightarrow+\infty }\frac{I_{max}(\cc_3)}{I_{W3}}=1$. The codebook $\cc_3$ asymptotically meets the Welch bound. This completes the proof.
\end{proof}

In Table \ref{table4}, we provide some explicit values of the parameters of the codebooks we proposed for some given $q$, and corresponding numerical data of the Welch bound for comparison. The numerical results show  that the codebooks $\cc_3$ asymptotically meet the Welch bound.

\begin{table}[!htbp]
 \newcommand{\tabincell}[2]{\begin{tabular}{@{}#1@{}}#2\end{tabular}}
 \caption{\footnotesize Parameters of the $(N_3,K_3)$ codebook of Section IV}
 \label{table4}
  \centering
 \setlength{\tabcolsep}{1mm}{
 \begin{tabular}{|c|c|c|c|c|c|}
   \hline
  $q$ & $N_3$ & $K_3$ & $I_{max}(\cc_3)$ & $I_{W3}$ & $\frac{I_{max}(\cc_3)}{I_{W3}}$ \\ \hline
 $3$& $18$& $4$& $0.7500$& $0.4537$& $1.6529$ \\ \hline
 $5$& $100$& $16$& $0.3125$& $0.2303$& $1.3570$ \\ \hline
 $13$& $2028$& $144$& $0.0903$& $0.0803$& $1.1237$ \\ \hline
 $49$& $115248$& $2304$& $0.0213$& $0.0206$& $1.0312$ \\ \hline
 $5^{3}$& $1937500$& $15376$& $0.0081$& $0.0080$& $1.0121$ \\ \hline
 $5^{4}$& $243750000$& $389376$& $0.0016$& $0.0016$& $1.0024$ \\ \hline
 $7^{4}$& $1.3830e+10$& $5760000$& $4.1684e-04$& $4.1658e-04$& $1.0006$ \\ \hline
 \end{tabular}}
 \end{table}

\subsection{The fourth construction of codebooks}
Let

$$
D_4:=\{(x,y,z)\in \Ff_{q}^*\times\Ff_q^*\times\Ff_q^*: z=xy\}.
$$
Then $\#D_4=(q-1)^2$.

The codebook $\cc_4$ is constructed as
$$\cc_4:=\{\frac{1}{q-1}(\varphi_i(x)\varphi_j(y)\chi_c(z))_{(x,y,z)\in D_4}:\varphi_i,\varphi_j\in \widehat{\Ff_q^*},\ c\in \Ff_q \}.$$
We can derive the following Theorem.

\begin{thm}\label{th41}
 $\cc_4$ is a codebook with $N_4=(q-1)^2q$, $K_4=(q-1)^2$ and $I_{max}(\cc_4)=\frac{q}{(q-1)^2}$.
\end{thm}

\begin{proof}
By the definition of $\cc_4$, it contains $(q-1)^2q$ codewords of length $(q-1)^2$. Then it is easy to see $N_4=(q-1)^2q$ and $K_4=(q-1)^2$. Let $\cF$ and $\cF'$ be any different codewords in $\cc_4$, $\cF=\frac{1}{q-1}(\varphi_s(x)\varphi_t(y)\chi_{c_1}(z))_{(x,y,z)\in D_4}$ and $\cF'=\frac{1}{q-1}(\varphi_s'(x)\varphi_t'(y)\chi_{c_2}(z))_{(x,y,z)\in D_4}$, where $\varphi_s,\varphi_t,\varphi_s',\varphi_t'\in \widehat{\Ff_q^*},c_1,c_2\in \Ff_q$. Then the correlation of $\cF$ and $\cF'$ is as follows.

\begin{eqnarray*}
&&K_4\cF\cF'^H\\
&=&\sum_{(x,y,z)\in D_4}\varphi_s(x)\varphi_t(y)\chi_{c_1}(z)\overline{\varphi_s'(x)\varphi_t'(y)\chi_{c_2}(z)}\\
&=&\sum_{(x,y,z)\in D_4}\varphi_s\overline{\varphi_s'}(x)\varphi_t\overline{\varphi_t'}(y)\chi((c_1-c_2)z)\\
&=&\sum_{x,y\in \Ff_q^*}\varphi(x)\varphi'(y)\chi(cxy) \ \ (\hbox{where $\varphi=\varphi_s\overline{\varphi_s'}$, $\varphi'=\varphi_t\overline{\varphi_t'}$, and $c=c_1-c_2$})\\
&=& \sum_{x\in \Ff_q^*}\varphi(x)\chi_c(x)\sum_{y\in \Ff_q^*}\varphi'(y)\chi_c(y)\\
&=&G(\varphi,\chi_c)G(\varphi',\chi_c).
\end{eqnarray*}

When $c=0$, since $\cF\neq\cF'$, at least one of $\varphi$ and $\varphi'$ is nontrivial, by Lemma \ref{gauss}, we have
$$K_4\cF\cF'^H=G(\varphi,\chi_c)G(\varphi',\chi_c)=0.$$

When $c\neq0$, by Lemma \ref{gauss} and \ref{gauss1}, we have

$$
 K_4\cF\cF'^H=G(\varphi,\chi_c)G(\varphi',\chi_c)=\left\{
            \begin{array}{ll}
              (-1)(-1), & \hbox{both $\varphi$ and $\varphi'$ are trivial,} \\
              (-1)\overline{\varphi'(c)}g(\varphi'),& \hbox{$\varphi$ is trivial, and $\varphi'$ is nontrivial,}\\
              (-1)\overline{\varphi(c)}g(\varphi),& \hbox{$\varphi$ is nontrivial, and $\varphi'$ is trivial,}\\
              \overline{\varphi\varphi'}(c)g(\varphi)g(\varphi'),& \hbox{both $\varphi$ and  $\varphi'$ are nontrivial.}
            \end{array}
          \right.
$$

Therefore, we have $$I_{max}(\cc_4)=max\{|\cF\cF'^H|:\cF,\cF'\in \cc_4, and\ \cF\neq\cF'\}=\frac{q}{K_4}=\frac{q}{(q-1)^2}.$$
\end{proof}

Using Theorem \ref{th41}, we can derive the ratio of $I_{max}(\cc_4)$ of the proposed codebooks to that of the MWBE codebooks and show the near-optimality of the proposed codebooks as in the following theorem.

\begin{thm}\label{th42}
Let $I_{W4}$ be the Welch bound equality, for the given $N_4$, $K_4$ in the current section. We have

$$\lim_{r\rightarrow\infty}\frac{I_{max}(\cc_4)}{I_{W4}}=1,$$
then the codebooks we proposed are near optimal.

\end{thm}
\begin{proof}
Note that $N_4=(q-1)^2 q$ and $K_4=(q-1)^2$. Then the corresponding Welch bound is
$$
I_{W4}=\sqrt{\frac{N_4-K_4}{(N_4-1)K_4}}=\sqrt{\frac{(q-1)^2 q-(q-1)^2}{((q-1)^2 q-1)(q-1)^2}}=\sqrt{\frac{q-1}{(q-1)^2q-1}},
$$
we have
$$\frac{I_{max}(\cc_4)}{I_{W4}}=\frac{q}{q-1}\sqrt{\frac{(q-1)^2q-1}{(q-1)^3}}.$$

It is obvious that $\lim_{q\rightarrow+\infty }\frac{I_{max}(\cc_4)}{I_{W4}}=1$. The codebook $\cc_4$ asymptotically meets the Welch bound. This completes the proof.
\end{proof}

In Table \ref{table5}, we provide some explicit values of the parameters of the codebooks we proposed for some given $q$, and corresponding numerical data of the Welch bound for comparison. The numerical results show  that the codebooks $\cc_4$ asymptotically meet the Welch bound.
\begin{table}[!htbp]
 \newcommand{\tabincell}[2]{\begin{tabular}{@{}#1@{}}#2\end{tabular}}
 \caption{\footnotesize Parameters of the $(N_4,K_4)$ codebook of Section IV}
 \label{table5}
 \centering
 \setlength{\tabcolsep}{1mm}{
 \begin{tabular}{|c|c|c|c|c|c|}
   \hline
  $q$ & $N_4$ & $K_4$ & $I_{max}(\cc_4)$ & $I_{W4}$ & $\frac{I_{max}(\cc_4)}{I_{W4}}$ \\ \hline
 $3$& $12$& $4$& $0.7500$& $0.4264$& $1.7589$ \\ \hline
 $5$& $80$& $16$& $0.3125$& $0.2250$& $1.3888$ \\ \hline
 $13$& $1872$& $144$& $0.0903$& $0.0801$& $1.1273$ \\ \hline
 $49$& $112896$& $2304$& $0.0213$& $0.0206$& $1.0314$ \\ \hline
 $5^{3}$& $1922000$& $15376$& $0.0081$& $0.0080$& $1.0121$ \\ \hline
 $5^{4}$& $243360000$& $389376$& $0.0016$& $0.0016$& $1.0024$ \\ \hline
 $7^{4}$& $1.3830e+10$& $5760000$& $4.1684e-04$& $4.1658e-04$& $1.0006$ \\ \hline

 \end{tabular}}
 \end{table}

\subsection{The fifth construction of codebooks}
Let

$$
D_5:=\{(x,y,z)\in \Ff_{q}^*\times\Ff_q^*\times\Ff_q^*: z=x(1-y)\}.
$$
Then $\#D_5=(q-1)(q-2)$.

The codebook $\cc_5$ is constructed as
$$\cc_5:=\{\frac{1}{\sqrt{(q-1)(q-2)}}(\chi_a(x)\varphi_i(y)\varphi_j(z))_{(x,y,z)\in D_5}:\varphi_i,\varphi_j\in \widehat{\Ff_q^*},\ a\in \Ff_q \}.$$
We can derive the following Theorem.

\begin{thm}\label{th51}
 $\cc_5$ is a codebook with $N_5=(q-1)^2q$, $K_5=(q-1)(q-2)$ and $I_{max}(\cc_5)=\frac{q}{(q-1)(q-2)}$.
\end{thm}

\begin{proof}
By the definition of $\cc_5$, it contains $(q-1)^2q$ codewords of length $(q-1)(q-2)$. Then it is easy to see $N_5=(q-1)^2q$ and $K_5=(q-1)(q-2)$. Let $\cF$ and $\cF'$ be any different codewords in $\cc_5$, $\cF=\frac{1}{\sqrt{(q-1)(q-2)}}(\chi_{a_1}(x)\varphi_s(y)\varphi_t(z))_{(x,y,z)\in D_5}$ and $\cF'=\frac{1}{\sqrt{(q-1)(q-2)}}(\chi_{a_2}(x)\varphi_s'(y)\varphi_t'(z))_{(x,y,z)\in D_5}$, where $\varphi_s,\varphi_t,\varphi_s',\varphi_t'\in \widehat{\Ff_q^*},a_1,a_2\in \Ff_q$. Then the correlation of $\cF$ and $\cF'$ is as follows.

\begin{eqnarray*}
&&K_5\cF\cF'^H\\
&=&\sum_{(x,y,z)\in D_5}\chi_{a_1}(x)\varphi_s(y)\varphi_t(z)\overline{\chi_{a_2}(x)\varphi_s'(y)\varphi_t'(z)}\\
&=&\sum_{(x,y,z)\in D_5}\chi((a_1-a_2)x)\varphi_s\overline{\varphi_s'}(y)\varphi_t\overline{\varphi_t'}(z)\\
&=&\sum_{x,y\in \Ff_q^*,y\neq1}\chi(ax)\varphi(y)\varphi'(x(1-y)) \ \ (\hbox{where $a=a_1-a_2$, $\varphi=\varphi_s\overline{\varphi_s'}$, and $\varphi'=\varphi_t\overline{\varphi_t'}$})\\
&=& \sum_{x\in \Ff_q^*}\chi(ax)\varphi'(x)\sum_{y\in \Ff_q^*,y\neq1}\varphi(y)\varphi'(1-y)\\
&=&G(\varphi',\chi_a)(J(\varphi,\varphi')-\varphi(0)\varphi'(1)-\varphi(1)\varphi'(0)).
\end{eqnarray*}

When $a=0$, since $\cF\neq\cF'$, at least one of $\varphi$ and $\varphi'$ is nontrivial, by Lemma \ref{gauss} and \ref{jacobi}, we have
$$
 K_5\cF\cF'^H=\left\{
            \begin{array}{ll}
              (-1)(q-1),& \hbox{$\varphi'$ is trivial, and $\varphi$ is nontrivial,}\\
              0 ,& \hbox{$\varphi'$ are nontrivial.}
            \end{array}
          \right.
$$

When $a\neq0$, by Lemma \ref{gauss} and \ref{jacobi}, we have

$$
 K_5\cF\cF'^H=\left\{
            \begin{array}{ll}
              (-1)(q-2), & \hbox{both $\varphi'$ and $\varphi$ are trivial,} \\
              (-1)(-1),& \hbox{$\varphi'$ is trivial, and $\varphi$ is nontrivial,}\\
              \overline{\varphi'}(a)g(\varphi')(-1),& \hbox{$\varphi'$ is nontrivial, and $\varphi$ is trivial,}\\
              \overline{\varphi'}(a)g(\varphi')J(\varphi,\varphi'),& \hbox{both $\varphi'$ and  $\varphi$ are nontrivial.}
            \end{array}
          \right.
$$

Therefore, we have $$I_{max}(\cc_5)=max\{|\cF\cF'^H|:\cF,\cF'\in \cc, and\ \cF\neq\cF'\}=\frac{q}{K_5}=\frac{q}{(q-1)(q-2)},$$
the maximal value obtained when all of $\varphi, \varphi'$ and $\varphi\varphi'$ are nontrivial.
\end{proof}

Using Theorem \ref{th51}, we can derive the ratio of $I_{max}(\cc_5)$ of the proposed codebooks to that of the MWBE codebooks and show the near-optimality of the proposed codebooks as in the following theorem.

\begin{thm}\label{th52}
Let $I_{W5}$ be the Welch bound equality, for the given $N_5$, $K_5$ in the current section. We have

$$\lim_{r\rightarrow\infty}\frac{I_{max}(\cc_5)}{I_{W5}}=1,$$
then the codebooks we proposed are near optimal.

\end{thm}
\begin{proof}
Note that $N_5=(q-1)^2 q$ and $K_5=(q-1)(q-2)$. Then the corresponding Welch bound is
$$
I_{W5}=\sqrt{\frac{N_5-K_5}{(N_5-1)K_5}}=\sqrt{\frac{(q-1)^2 q-(q-1)(q-2)}{((q-1)^2 q-1)(q-1)(q-2)}}=\sqrt{\frac{q^2-2q+2}{(q(q-1)^2-1)(q-2)}},
$$
we have
$$\frac{I_{max}(\cc_5)}{I_{W5}}=\frac{q}{q-1}\sqrt{\frac{q(q-1)^2-1}{(q^2-2q+2)(q-2)}}.$$

It is obvious that $\lim_{q\rightarrow+\infty }\frac{I_{max}(\cc_5)}{I_{W5}}=1$. The codebook $\cc_5$ asymptotically meets the Welch bound. This completes the proof.
\end{proof}

In Table \ref{table6}, we provide some explicit values of the parameters of the codebooks we proposed for some given $q$, and corresponding numerical data of the Welch bound for comparison. The numerical results show  that the codebooks $\cc_5$ asymptotically meet the Welch bound.
\begin{table}[!htbp]
 \newcommand{\tabincell}[2]{\begin{tabular}{@{}#1@{}}#2\end{tabular}}
 \caption{\footnotesize Parameters of the $(N_5,K_5)$ codebook of Section IV}
 \label{table6}
 \centering
 \setlength{\tabcolsep}{1mm}{
 \begin{tabular}{|c|c|c|c|c|c|}
   \hline
  $q$ & $N_5$ & $K_5$ & $I_{max}(\cc_5)$ & $I_{W5}$ & $\frac{I_{max}(\cc_5)}{I_{W5}}$ \\ \hline
 $3$& $12$& $2$& $1.5000$& $0.6742$& $2.2249$ \\ \hline
 $5$& $80$& $12$& $0.4167$& $0.2678$& $1.5557$ \\ \hline
 $13$& $1872$& $132$& $0.0985$& $0.0839$& $1.1733$ \\ \hline
 $49$& $112896$& $2256$& $0.0217$& $0.0208$& $1.0421$ \\ \hline
 $5^{3}$& $1922000$& $15252$& $0.0082$& $0.0081$& $1.0162$ \\ \hline
 $5^{4}$& $243360000$& $388752$& $0.0016$& $0.0016$& $1.0032$ \\ \hline
 $7^{4}$& $1.3830e+10$& $5757600$& $4.1701e-04$& $4.1667e-04$& $1.0008$ \\ \hline

 \end{tabular}}
 \end{table}

\subsection{The sixth construction of codebooks}
Let

$$
D_6:=\{(x,y,z)\in \Ff_{q}^*\times\Ff_q^*\times\Ff_q^*: z=(1-x)(1-y)\}.
$$
Then $\#D_6=(q-2)^2$.

The codebook $\cc_6$ is constructed as
$$\cc_6:=\{\frac{1}{q-2}(\varphi_i(x)\varphi_j(y)\varphi_k(z))_{(x,y,z)\in D_6}:\varphi_i,\varphi_j, \varphi_k\in \widehat{\Ff_q^*} \}.$$
We can derive the following Theorem.

\begin{thm}\label{th61}
 $\cc_6$ is a codebook with $N_6=(q-1)^3$, $K_6=(q-2)^2$ and $I_{max}(\cc_6)=\frac{q}{(q-
 2)^2}$.
\end{thm}

\begin{proof}
By the definition of $\cc_6$, it contains $(q-1)^3$ codewords of length $(q-2)^2$. Then it is easy to see $N_6=(q-1)^3$ and $K_6=(q-2)^2$. Let $\cF$ and $\cF'$ be any different codewords in $\cc_6$, $\cF=\frac{1}{q-2}(\varphi_s(x)\varphi_u(y)\varphi_v(z))_{(x,y,z)\in D_6}$ and $\cF'=\frac{1}{q-2}(\varphi_s'(x)\varphi_u'(y)\varphi_v'(z))_{(x,y,z)\in D_6}$, where $\varphi_s,\varphi_u,\varphi_v,\varphi_s',\varphi_u',\varphi_v'\in \widehat{\Ff_q^*}$. Then the correlation of $\cF$ and $\cF'$ is as follows.

\begin{eqnarray*}
&&K\cF\cF'^H\\
&=&\sum_{(x,y,z)\in D_6}\varphi_s(x)\varphi_u(y)\varphi_v(z)\overline{\varphi_s'(x)\varphi_u'(y)\varphi_v'(z)}\\
&=&\sum_{(x,y,z)\in D_6}\varphi_s\overline{\varphi_s'}(x)\varphi_u\overline{\varphi_u'}(y)\varphi_v\overline{\varphi_v'}(z)\\
&=&\sum_{x,y\in \Ff_q^*,x\neq1,y\neq 1}\varphi_i(x)\varphi_j(y)\varphi_k((1-x)(1-y)) \ \ (\hbox{where $\varphi_i=\varphi_s\overline{\varphi_s'}$, $\varphi_j=\varphi_u\overline{\varphi_u'}$, and $\varphi_k=\varphi_v\overline{\varphi_v'}$})\\
&=& \sum_{x\in \Ff_q^*,x\neq 1}\varphi_i(x)\varphi_k(1-x)\sum_{y\in \Ff_q^*,y\neq1}\varphi_j(y)\varphi_k(1-y)\\
&=&(J(\varphi_i,\varphi_k)-\varphi_i(0)\varphi_k(1)-\varphi_i(1)\varphi_k(0))(J(\varphi_j,\varphi_k)-\varphi_j(0)\varphi_k(1)
-\varphi_j(1)\varphi_k(0)).
\end{eqnarray*}

When $\varphi_k$ is nontrivial, since $\cF\neq\cF'$, at least one of $\varphi_i$ and $\varphi_j$ is nontrivial, by Lemma \ref{jacobi}, we have
$$
 K_6\cF\cF'^H=\left\{
            \begin{array}{ll}
              (-1)(q-2),& \hbox{$\varphi_i$ is trivial, and $\varphi_j$ is nontrivial,}\\
              (-1)(q-2),& \hbox{$\varphi_i$ is nontrivial, and $\varphi_j$ is trivial,}\\
              (-1)(-1) ,& \hbox{both $\varphi_i$ and $\varphi_j$ are nontrivial.}
            \end{array}
          \right.
$$

When $\varphi_k$ is nontrivial, by Lemma \ref{jacobi}, we have

$$
 K_6\cF\cF'^H=\left\{
            \begin{array}{ll}
              (-1)(-1), & \hbox{both $\varphi_i$ and $\varphi_j$ are trivial,} \\
              (-1)J(\varphi_j,\varphi_k),& \hbox{$\varphi_i$ is trivial, and $\varphi_j$ is nontrivial,}\\
               (-1)J(\varphi_i,\varphi_k),& \hbox{$\varphi_i$ is nontrivial, and  $\varphi_j$ is trivial,}\\
              J(\varphi_i,\varphi_k)J(\varphi_j,\varphi_k),& \hbox{both $\varphi_i$ and  $\varphi_j$ are nontrivial.}
            \end{array}
          \right.
$$

Therefore, we have $$I_{max}(\cc_6)=max\{|\cF\cF'^H|:\cF,\cF'\in \cc, and\ \cF\neq\cF'\}=\frac{q}{K_6}=\frac{q}{(q-2)^2},$$
the maximal value obtained when all of $\varphi_i,\varphi_j,\varphi_k,\varphi_i\varphi_k$  and $\varphi_j\varphi_k$ are nontrivial.
\end{proof}

Using Theorem \ref{th61}, we can derive the ratio of $I_{max}(\cc_6)$ of the proposed codebooks to that of the MWBE codebooks and show the near-optimality of the proposed codebooks as in the following theorem.

\begin{thm}\label{th62}
Let $I_{W6}$ be the Welch bound equality, for the given $N_6$, $K_6$ in the current section. We have

$$\lim_{r\rightarrow\infty}\frac{I_{max}(\cc_6)}{I_{W6}}=1,$$
then the codebooks we proposed are near optimal.

\end{thm}
\begin{proof}
Note that $N_6=(q-1)^3$ and $K_6=(q-2)^2$. Then the corresponding Welch bound is
$$
I_{W6}=\sqrt{\frac{N_6-K_6}{(N_6-1)K_6}}=\sqrt{\frac{(q-1)^3-(q-2)^2}{((q-1)^3-1)(q-2)^2}}=\frac{1}{q-2}\sqrt{\frac{q^3-4q^2+7q-5}{q^3-3q^2+3q-2}},
$$
we have
$$\frac{I_{max}(\cc_6)}{I_{W6}}=\frac{q}{q-2}\sqrt{\frac{q^3-3q^2+3q-2}{q^3-4q^2+7q-5}}.$$

It is obvious that $\lim_{q\rightarrow+\infty }\frac{I_{max}(\cc_6)}{I_{W6}}=1$. The codebook $\cc$ asymptotically meets the Welch bound. This completes the proof.
\end{proof}

In Table \ref{table7}, we provide some explicit values of the parameters of the codebooks we proposed for some given $q$, and corresponding numerical data of the Welch bound for comparison. The numerical results show  that the codebooks $\cc_6$ asymptotically meet the Welch bound.
\begin{table}[!htbp]
 \newcommand{\tabincell}[2]{\begin{tabular}{@{}#1@{}}#2\end{tabular}}
 \caption{\footnotesize Parameters of the $(N_6,K_6)$ codebook of Section IV}
 \label{table7}
 \centering
 \setlength{\tabcolsep}{1mm}{
 \begin{tabular}{|c|c|c|c|c|c|}
   \hline
  $q$ & $N_6$ & $K_6$ & $I_{max}(\cc_6)$ & $I_{W6}$ & $\frac{I_{max}(\cc_6)}{I_{W6}}$ \\ \hline
 $3$& $8$& $1$& $3$& $1$& $3$ \\ \hline
 $5$& $64$& $9$& $0.5556$& $0.3115$& $1.7838$ \\ \hline
 $13$& $1728$& $121$& $0.1074$& $0.0877$& $1.2251$ \\ \hline
 $49$& $110592$& $2209$& $0.0222$& $0.0211$& $1.0531$ \\ \hline
 $5^{3}$& $1906624$& $15129$& $0.0083$& $0.0081$& $1.0203$ \\ \hline
 $5^{4}$& $242970624$& $388129$& $0.0016$& $0.0016$& $1.0040$ \\ \hline
 $7^{4}$& $1.3824e+10$& $5755201$& $4.1719e-04$& $4.1675e-04$& $1.0010$ \\ \hline

 \end{tabular}}
 \end{table}

\section{Another family  of codebooks}
Based on the six constructions of codebooks in Section 3, more new codebooks can be derived, which are also near optimal.

Let $\mathcal{E}_n$ denote the set formed by the standard basis of the n-dimensional Hilbert space, combining with the preceding six constructions, we get the following result.

\begin{thm}\label{th71}
Let $\cc'_i=\cc_i\cup\mathcal{E}_{K_i}$. Then the codebooks $\cc'_i$ are all near optimal, $i=1,2,3,4,5,6$, and the parameters of the new codebooks are as follows:

$N_i'=N_i+K_i$, $K_i'=K_i$ and $I_{max}(\cc_i')=I_{max}(\cc_i)$. Specically,

(1) $N_1'=N_1+K_1=q^3+q^2$, $K_1'=K_1=q^2$ and $I_{max}(\cc_1')=I_{max}(\cc_1)=\frac{1}{q}$;

(2) $N_2'=N_2+K_2=(q^2-1)$, $K_2'=K_2=(q-1)q$ and $I_{max}(\cc_2')=I_{max}(\cc_2)=\frac{1}{q-1}$;

(3) $N_3'=N_3+K_3=q^3-2q+1$, $K_3'=K_3=(q-1)^2$ and $I_{max}(\cc_3')=I_{max}(\cc_3)=\frac{q}{(q-1)^2}$;

(4) $N_4'=N_4+K_4=q^3-q^2-q+1$, $K_4'=K_4=(q-1)^2$ and $I_{max}(\cc_4')=I_{max}(\cc_4)=\frac{q}{(q-1)^2}$;

(5) $N_5'=N_5+K_5=q^3-q^2-2q+2$, $K_5'=K_5=(q-1)(q-2)$ and $I_{max}(\cc_5')=I_{max}(\cc_5)=\frac{q}{(q-1)(q-2)}$;

(6) $N_5'=N_6+K_6=q^3-2q^2-q+3$, $K_6'=K_6=(q-2)^2$ and $I_{max}(\cc_5')=I_{max}(\cc_5)=\frac{q}{(q-2)^2}$.

\end{thm}

\begin{proof}
We only prove the case when $i=2$, the other cases can be proved as a similar way.  It is easy to see $N_2'=N_2+K_2=(q^2-1)q$ and $K_2'=K_2=(q-1)q$. Let $\cF$ and $\cF'$ be any different codewords in $\cc_2'$, the correlation of $\cF$ and $\cF'$ can be discussed in the following cases.

Case 1: If $\cF$, $\cF'\in\ \cc_2$, by Theorem \ref{th21}, we have
$$
 max\{|\cF\cF'^H|:\cF,\cF'\in \cc_2, and\ \cF\neq\cF'\}=\frac{1}{q-1}
$$

Case 2: If $\cF\in \cc_2$ and $\cF'\in\mathcal{E}_{(q-1)q}$ (or $\cF'\in \cc_2$ and $\cF\in\mathcal{E}_{(q-1)q}$), it is easy to see
$$|\cF\cF'^H|=\frac{1}{\sqrt{(q-1)q}}.$$

Case 3: If $\cF$ and $\cF'\in \mathcal{E}_{(q-1)q}$, it is obvious that $\cF\cF'^H=0.$

From the above three cases, we have $$I_{max}(\cc_2')=\frac{1}{q-1}.$$

Note that $N_2'=(q^2-1)q$ and $K_2'=(q-1)q$, then the corresponding Welch bound is
$$
I_W=\sqrt{\frac{N_2'-K_2'}{(N_2'-1)K_2'}}=\sqrt{\frac{q}{q^3-q-1}},
$$
we have
$$\frac{I_{max}(\cc_2')}{I_W}=\sqrt{\frac{q^3-q-1}{(q-1)^2q}}.$$

It is obvious that $\lim_{q\rightarrow+\infty }\frac{I_{max}(\cc_2')}{I_W}=1$. The codebook $\cc_2'$ asymptotically meets the Welch bound. This completes the proof.
\end{proof}

\section{Concluding remarks}

In this paper, we presented six new constructions of codebooks asymptotically achieve the Welch bounds with additive characters, multiplicative characters and character sums of finite fields. Actually, the first construction in our paper is  equivalent to the measurement matrix in \cite{MM}.
 The advantage of our construction is that it can be generalized naturally to  the other five constructions of codebooks which are also near optimal.

\end{document}